\documentclass[11pt,a4paper]{article}

\usepackage{amsmath,amssymb,amsfonts}
\usepackage{amsthm}
\usepackage{graphicx}
\usepackage{textcomp}
\usepackage{braket}
\usepackage{xcolor}
\usepackage{a4wide}
\usepackage{authblk}
\usepackage{tikz}

\usepackage{hyperref}

\usepackage[ruled,vlined, linesnumbered]{algorithm2e}
\newtheorem{thm}{Theorem}
\newtheorem{cor}{Corollary}

\newtheorem{prop}{Proposition}

\DeclareMathOperator{\E}{\mathbb{E}}

\begin{document}

%*******TITLE AND AUTHORS*******************************************
\title{Robust optimization with belief functions} 

\author[1]{Marc Goerigk}
\author[2]{Romain Guillaume}
\author[3]{Adam Kasperski\footnote{Corresponding author}}
\author[3]{Pawe{\l} Zieli\'nski}

\affil[1]{Network and Data Science Management, University of Siegen, Siegen, Germany,
           \texttt{marc.goerigk@uni-siegen.de}}
  \affil[2]{Universit{\'e} de Toulouse-IRIT Toulouse, France, \texttt{romain.guillaume@irit.fr}}
\affil[3]{
Wroc{\l}aw  University of Science and Technology, Wroc{\l}aw, Poland\\
            \texttt{\{adam.kasperski,pawel.zielinski\}@pwr.edu.pl}}

\date{}

\maketitle

 \begin{abstract}
 In this paper, an optimization problem with uncertain objective function coefficients is considered. The uncertainty is specified by providing a discrete scenario set, containing possible realizations of the objective function coefficients. The concept of belief function in the traditional and possibilistic setting is applied to define a set of admissible probability distributions over the scenario set. The generalized Hurwicz criterion is then used to compute a solution. In this paper, the complexity of the resulting problem is explored. Some exact and approximation methods of solving it are proposed.
 \end{abstract}
\noindent \textbf{Keywords}: robust optimization, Hurwicz criterion, belief function, possibility theory

\section{Introduction}

In this paper,
we wish to investigate a version
of
 the following generic optimization problem under uncertainty in  the objective function coefficients:
\begin{equation}
\label{pf}
\begin{array}{lll}
	\begin{array}{llll}
		\min  & f(\pmb{x},\pmb{c}) \\
			& \pmb{x}\in \mathbb{X},
	\end{array}
\end{array}
\end{equation}
where $\pmb{c}\in \mathbb{R}^l$ is a vector of the objective function coefficients and $\mathbb{X}\subseteq \mathbb{R}^n$ is a set of feasible solutions. We will assume that $\mathbb{X}$ is not empty and compact, $f$ takes only nonnegative values,
$f:\mathbb{X}\times  \mathbb{R}^l\rightarrow  \mathbb{R}_{+}$,
 and  attains a minimum in $\mathbb{X}$ for each fixed vector of coefficients $\pmb{c}$.
If $\mathbb{X}$ is a polytope, described by a system of linear constraints, and $f$ is linear, then we get the class of linear programming problems. In a more general case, if $\mathbb{X}$ is a convex set and $f$ is a convex function, then we get the class of convex optimization problems (see, e.g.,~\cite{BV04}). If $\mathbb{X}\subseteq \{0,1\}^n$, then~(\ref{pf}) is a combinatorial problem.  Depending on the structure of $\mathbb{X}$ and the definition of the cost function $f$, the problem~(\ref{pf}) can be solved in polynomial time or becomes NP-hard.

 In many practical applications, the vector of the objective function coefficients $\pmb{c}$ is uncertain. Namely, $\tilde{\pmb{c}}=(\tilde{c}_1,\dots,\tilde{c}_l)$ is then a random vector in $\mathbb{R}^l$, whose probability distribution is known, partially known or completely unknown. Accordingly, the cost $f(\pmb{x},\tilde{\pmb{c}})$ of a given solution $\pmb{x}\in \mathbb{X}$ is a random variable in $\mathbb{R}_+$.
Let $\mathcal{U}$ be a \emph{set of scenarios} containing possible realizations of $\tilde{\pmb{c}}$ called \emph{scenarios}.  In this paper, we will assume that $\mathcal{U}=\{\pmb{c}_1,\dots,\pmb{c}_K\}\subseteq \mathbb{R}^l$ contains a finite number of $K\geq 1$, explicitly listed, scenarios. Scenarios in $\mathcal{U}$ may correspond to possible states of the world or can be a result of sampling of the random vector~$\tilde{\pmb{c}}$.
We will use  $\mathcal{P}(\mathcal{U})$ to denote the set of all discrete probability distributions in~$\mathcal{U}$, that is, $\mathcal{P}(\mathcal{U})=\{\pmb{p}\in [0,1]^K: \sum_{k\in[K]} p_k=1\}$. If some additional knowledge in $\mathcal{P}(\mathcal{U})$ is available, then one can use it to provide a set of \emph{admissible probability distributions} $\mathcal{P}\subseteq \mathcal{P}(\mathcal{U})$, which is a restriction of $\mathcal{P}(\mathcal{U})$. For example, one can add some bounds $\underline{p}_i\leq p_i \leq \overline{p}_i$, $i\in [K]$ (we use the shortcut $[K]$ to denote the set $\{1,\dots,K\}$), on the probabilities~\cite{O12}, or some bounds on the mean or covariance matrix for $\tilde{\pmb{c}}$, which leads to various descriptions of  $\mathcal{P}$ (see, e.g.,~\cite{DY10, WKS14}).  Fix $\alpha\in [0,1]$ and define
$${\rm H}_{\alpha}(\pmb{x})= \alpha  {\rm \overline{E}}[f(\pmb{x},\tilde{\pmb{c}})] + (1-\alpha)  {\rm \underline{E}}[f(\pmb{x},\tilde{\pmb{c}})], $$
where
\begin{align*}
  &{\rm \overline{E}}[f(\pmb{x},\tilde{\pmb{c}})]=\sup_{\pmb{p}\in \mathcal{P}} \E_{\pmb{p}}[f(\pmb{x},\tilde{\pmb{c}})]=\sup_{\pmb{p}\in \mathcal{P}} \sum_{k\in [K]} p_k f(\pmb{x},\pmb{c}_k),\\
 & {\rm \underline{E}}[f(\pmb{x},\tilde{\pmb{c}})]=\inf_{\pmb{p}\in \mathcal{P}} \E_{\pmb{p}}[f(\pmb{x},\tilde{\pmb{c}})]=\inf_{\pmb{p}\in \mathcal{P}} \sum_{k\in [K]} p_k f(\pmb{x},\pmb{c}_k),
\end{align*}
and $\E_{\pmb{p}}[\cdot]$ is the expected value with respect to the probability distribution~$\pmb{p}$.
The quantity ${\rm H}_{\alpha}$ is called a \emph{generalized Hurwicz criterion,} and $\overline{\rm E}[\cdot]$, $\underline{\rm E}[\cdot ]$ are called the \emph{upper and lower expected values} of $f(\pmb{x},\tilde{\pmb{c}})$, respectively, representing the largest and  smallest expected solution costs for the probability distributions in $\mathcal{P}$ (see, e.g.,~\cite{TD20}).
The parameter $\alpha\in [0,1]$ is called a \emph{pessimism-optimism degree} and models a decision maker's risk aversion.

In this paper, we consider the following optimization problem:
\begin{equation}
\label{pfh}
\begin{array}{llll}
	\min &{\rm H}_{\alpha}(\pmb{x}) \\
	& \pmb{x}\in \mathbb{X}.
\end{array}
\end{equation}
Hence, we seek a feasible solution that minimizes the generalized Hurwicz criterion for a specified value of the pessimism-optimism degree $\alpha\in [0,1]$.

Let us focus on some special cases of~(\ref{pfh}).
 If $\mathcal{P}=\mathcal{P}(\mathcal{U})$, then ${\rm \overline{E}}[f(\pmb{x},\tilde{\pmb{c}})]=\max_{k\in [K]} f(\pmb{x},\pmb{c}_k)$. Indeed, in this case the worst probability distribution assigns the probability equal to~1 to scenario $\pmb{c}_k$ maximizing the cost of $\pmb{x}$. Similarly,
 ${\rm \underline{E}}[f(\pmb{x},\tilde{\pmb{c}})]=\min_{k\in [K]} f(\pmb{x},\pmb{c}_k)$. Problem~(\ref{pfh}) reduces then to minimizing the traditional Hurwicz criterion under scenario set $\mathcal{U}$:
 \begin{equation}
\label{pfhur}
\begin{array}{llll}
	\min & \displaystyle \alpha\max_{k\in [K]} f(\pmb{x},\pmb{c}_k)+(1-\alpha)\min_{k\in [K]} f(\pmb{x},\pmb{c}_k) \\
	& \pmb{x}\in \mathbb{X}.
\end{array}
\end{equation}
If $\alpha=1$ in~(\ref{pfhur}), then we get the \emph{robust min-max problem}, discussed for example, in~\cite{KY97}.

If $\mathcal{P}$ contains only one probability distribution, $\pmb{p}$, then~(\ref{pfh}) is equivalent to minimizing the expected solution cost under $\pmb{p}$, i.e.
  \begin{equation}
\label{pfexp}
\begin{array}{llll}
	\min & \E_{\pmb{p}}[f(\pmb{x},\tilde{\pmb{c}})]=\displaystyle \sum_{k\in [K]} p_k f(\pmb{x},\pmb{c}_k) \\
	& \pmb{x}\in \mathbb{X}.
\end{array}
\end{equation}
 
 In the traditional robust optimization approach (see, e.g.,~\cite{BN09, KY97}) no additional information in scenario set $\mathcal{U}$ is provided. This is equivalent to saying that any probability distribution in $\mathcal{P}(\mathcal{U})$ for $\tilde{\pmb{c}}$ is admissible. The problem of type~(\ref{pfhur}) can  then  be solved. However, in many practical applications, some additional knowledge in $\mathcal{U}$ can be available. In particular, some evidence can indicate that scenarios from one subset of $\mathcal{U}$ are more likely to occur than scenarios from another subset. This knowledge can be naturally taken into account by using Dempster and Shafer \emph{theory of evidence}~\cite{D67, SH76}. The key concept of this theory is nonnegative masses assigned to the subsets of scenarios. Given a subset of scenarios $A\subseteq \mathcal{U}$, the mass assigned to $A$ can be interpreted as a \emph{subjective probability} of the event $A$, which supports the claim that the true scenario will belong to $A$. The masses provided induce a \emph{belief function} in $2^{\mathcal{U}}$, which constitutes a lower bound on the probability distribution in $\mathcal{U}$. This fact can be used to define the set $\mathcal{P}\subseteq \mathcal{P}(\mathcal{U})$ of admissible probability distributions for $\tilde{\pmb{c}}$ and   then apply the solution concept~(\ref{pfh}). The concept of belief function can be extended by using \emph{possibility theory} (see, e.g.,~\cite{DP88}). That is, a mass can be assigned to a fuzzy subset of scenarios. The membership function of such a fuzzy subset is a \emph{possibility distribution} in $\mathcal{U}$. The belief  of an event $A\subseteq \mathcal{U}$ is then the expected necessity of $A$. By using a decomposition of  fuzzy sets into $\lambda$-cuts~\cite{DP90consonant}, we can reduce the possibilistic model of uncertainty into a model based on the traditional theory of evidence.
 
 This paper is an extended version of the conference papers~\cite{GKZ21, GKZ22}. We generalize the results to a wider class of optimization problems. We strengthen the negative complexity result shown in~\cite{GKZ21}. We propose another approximation algorithm for~(\ref{pfh}), which is an alternative to the one constructed in~\cite{GKZ21}. Finally, we show how to solve the problem with a possibilistic model of uncertainty, which remained open in~\cite{GKZ21}.
 
This paper is organized as follows. In Section~\ref{secd}, we recall basic notions of the theory of evidence, namely the notions of the mass and belief functions. In Section~\ref{secrob}, we apply them to define the set $\mathcal{P}\subseteq \mathcal{P}(\mathcal{U})$ of admissible probability distributions in $\mathcal{U}$. We then consider the optimization problem~(\ref{pfh}) with $\mathcal{P}$. In Section~\ref{seccompl} we investigate the complexity of the problem~(\ref{pfh}). We show that it is NP-hard for each $\alpha\in [0,1)$ even in the very restrictive case when $\mathbb{X}$ is a polytope in $[0,1]^n$ and the function $f$ is linear. In Section~\ref{secsol} we propose a mixed integer programming problem for solving~(\ref{pfh}). We also identify some special cases of the problem which can be solved in polynomial time. In Section~\ref{secappr} we construct  approximation algorithms for~(\ref{pfh}). Finally, in Section~\ref{secpos} we extend the model of uncertainty to the possibilistic (fuzzy) case.

\section{Belief functions and imprecise probabilities}
\label{secd}
In this section, we recall the concept of a belief function and its relationships with imprecise probabilities (more details can be found in~\cite{D67, SH76, TD20}).
Let $\Omega=\{\omega_1,\dots,\omega_K\}$ be a  finite non-empty set of states. A \emph{mass function} is a mapping $m$ from $2^{\Omega}$ to $[0,1]$ such that
$$\sum_{A\subseteq \Omega} m(A)=1$$
and $m(\emptyset)=0$. Given $A$, $m(A)$ is available evidence that supports the claim that the actual state belongs to $A$. 
Set $A\subseteq \Omega$ such that $m(A)>0$ is called a \emph{focal set}.
The mass function $m$ induces the following \emph{belief function}:
\begin{equation}
\label{belf}
{\rm Bel}(A)=\sum_{B\subseteq A} m(B).
\end{equation}
If all focal sets are singletons, then ${\rm Bel}$ is a probability distribution in $\Omega$, so it provides  complete probabilistic information in~$\Omega$. On the other hand, if there is only one focal set, say $E$, then ${\rm Bel}$ is a \emph{logical measure}, describing imprecise information in $\Omega$. We only know for sure that event $E$ will occur and nothing more. In the extreme case when $E=\Omega$, there is complete uncertainty in $\Omega$

A probability distribution ${\rm P}$ in $\Omega$ is \emph{compatible with} ${\rm Bel}$ if ${\rm P}(A)\geq {\rm Bel}(A)$ for each event $A\subseteq \Omega$ (see, e.g.,~\cite{TD20}). Therefore, ${\rm Bel}(A)$ is a lower bound on the probability ${\rm P}(A)$.
If ${\rm Bel}$ is induced by the mass function $m$, then $\mathcal{P}(m)$ denotes the set of all probability distributions compatible with ${\rm Bel}$
\begin{equation}
\label{pm}
\mathcal{P}(m)=\{{\rm P}\in \mathcal{P}(\Omega): {\rm P}(A)\geq {\rm Bel}(A), A\subseteq [K]\},
\end{equation}
where $\mathcal{P}(\Omega)$ is the set of all probability distributions in $\Omega$.
 As $\Omega$ is finite, the set $\mathcal{P}(m)$ can be described by the following system of linear constraints:
\begin{equation}
\label{consp}
 	\begin{array}{llll}
		\displaystyle \sum_{k\in A} p_k\geq {\rm Bel}(A) & A\subset [K]\\
		p_1+\dots+p_K=1\\
		p_k\geq 0 & k\in [K],
	\end{array}
\end{equation}
where $p_k$ is a probability that state $\omega_k$ will occur. In Section~\ref{secpos} we will study a generalization of the belief function~(\ref{belf}). In particular, we will show that~(\ref{consp}) has at least one solution, so $\mathcal{P}(m)$ is nonempty.

%
% It is well known that  the belief function ${\rm Bel}$ is supermodular (see, e.g.,~\cite{KW11}). In this case, system~(\ref{consp}) describes the core of a convex cooperative game.  Because the cores of such games are nonempty~\cite{SH71}, there is at least one probability distribution compatible with~${\rm Bel}$ and $\mathcal{P}(m)\neq \emptyset$.

\section{Robust Optimization with Belief Functions (ROBF)}
\label{secrob}

In this section, we will use the concept of belief function described in Section~\ref{secd} to model the uncertainty in the objective function of~(\ref{pf}). Recall that the vector of the objective function coefficients $\tilde{\pmb{c}}$ can take one of a finite set of values in $\mathcal{U}=\{\pmb{c}_1,\dots,\pmb{c}_K\}\subseteq \mathbb{R}^l$, $K\geq 1$. In order to apply the concepts from Section~\ref{secd}, we will identify scenarios with states $\omega_1,\dots,\omega_K$, i.e. $\Omega=\mathcal{U}$.  Let $m: 2^\mathcal{U}\rightarrow [0,1]$ be a mass function in $\mathcal{U}$. 
To simplify notation, we will identify each subset of scenarios $A\subseteq \mathcal{U}$ with the set of indices of the elements of~$A$. Therefore, $m$ can be equivalently defined as a mass function in the power set $2^{[K]}$ of the indices, i.e. $m: 2^{[K]}\rightarrow [0,1]$. Let
$$\mathcal{F}=\{F\subseteq [K]: m(F)>0\}=\{F_1,\dots,F_{\ell}\}$$
be the set of all focal sets of $[K]$ and $z(\mathcal{F})=\max_{F\in \mathcal{F}} |F|$ be the size of the largest focal set. The mass function $m$ induces the belief function ${\rm Bel}$ in $\mathcal{U}$, which in turn provides us a set of $\mathcal{P}(m)$ of admissible probability distributions in $\mathcal{U}$. The set $\mathcal{P}(m)$ is described by~(\ref{consp}), where $p_k$ is now the probability that scenario $\pmb{c}_k$ will occur. The following equations are true (see, e.g.,~\cite{TD20}):
\begin{align}
\overline{\rm E}[f(\pmb{x},\tilde{\pmb{c}})]&=\sup_{{\rm P}\in \mathcal{P}(m)}  \E_{{\rm P}}[f(\pmb{x},\tilde{\pmb{c}})]=\sum_{F\in \mathcal{F}}m(F)\max_{k\in F}f(\pmb{x},\pmb{c}_k),\label{ue}\\
\underline{\rm E}[f(\pmb{x},\tilde{\pmb{c}})]&=\inf_{{\rm P}\in \mathcal{P}(m)} \E_{{\rm P}}[f(\pmb{x},\tilde{\pmb{c}})]=\sum_{F\in \mathcal{F}}m(F)\min_{k\in F}f(\pmb{x},\pmb{c}_k)\label{le}.
\end{align}
A proof of~(\ref{ue}) and~(\ref{le}) is shown in Appendix~\ref{app:a} for completeness. Using~(\ref{ue}) and~(\ref{le}) we can represent the problem~(\ref{pfh}) as the following \emph{robust opimization problem with belief functions}:
\begin{equation}
\label{pfh1}
(\text{ROBF})\;\;
\begin{array}{lll}
\min & 
\displaystyle {\rm H}_{\alpha}(\pmb{x})= \sum_{F\in \mathcal{F}}m(F)\left(\alpha\max_{k\in F}f(\pmb{x},\pmb{c}_k) + (1-\alpha)\min_{k\in F}f(\pmb{x},\pmb{c}_k)\right) \\
& \pmb{x}\in \mathbb{X}
\end{array}
\end{equation}
Let us focus on the first complexity issue that can arise while solving~(\ref{pfh1}). In general, specifying the mass function $m$ may require providing up to $2^K$ numbers, which can be intractable for larger $K$. In practice, we can assume that the number of focal sets $|\mathcal{F}|$ is bounded by a polynomial in $K$.  Alternatively, one can also provide a rule for assigning the masses to subsets of $K$, which leads to a reformulation of~(\ref{pfh1}) of a polynomial size in $n$ and $K$ (i.e. in the size of the problem). 

We now  show several special cases of~(\ref{pfh1}), which have been already discussed in the literature. If $m([K])=1$, then we get the problem with the  Hurwicz criterion~(\ref{pfhur}), which reduces to the robust min-max problem when additionally $\alpha=1$.
On the other hand, if $z(\mathcal{F})=1$ (i.e. all focal sets are singletons), then~(\ref{pfh1}) reduces to minimizing the expected solution cost for the probability distribution $p_k=m(\{k\})$, $k\in[K]$, i.e. to the problem~(\ref{pfexp}).

Let us describe yet another special case of~(\ref{pfh1}). In robust optimization with scenario set $\mathcal{U}$ the \emph{Ordered Weighted Averaging} criterion (OWA for short), proposed in~\cite{YA88}, is commonly used (see, e.g.,~\cite{CG15,KZ15,OO12}). Let $\pmb{w}=(w_1,\dots,w_K)$ be a vector of nonnegative weights that sum up to~1. Given a solution $\pmb{x}\in \mathbb{X}$, let $\sigma$ be a permutation of $[K]$ such that $f(\pmb{x},\pmb{c}_{\sigma(1)})\geq \dots\geq f(\pmb{x},\pmb{c}_{\sigma(K)})$, then 
$${\rm Owa}_{\pmb{w}}(\pmb{x})=\sum_{k\in [K]} w_k f(\pmb{x},\pmb{c}_{\sigma(k)}).$$
It is easy to see that ${\rm Owa}_{\pmb{w}}(\pmb{x})=\max_{k\in [K]} f(\pmb{x},\pmb{c}_k)$, when $\pmb{w}=(1,0,\dots,0)$ and ${\rm Owa}_{\pmb{w}}(\pmb{x})=\min_{k\in [K]} f(\pmb{x},\pmb{c}_k)$, when $\pmb{w}=(0,0,\dots,1)$.

Define $m(A)=\frac{2}{K(K-1)}$ for each $A\subseteq [K]$ such that $|A|=2$ and $m(A)=0$, otherwise. Hence, all focal sets 
are
cardinality of~2 and they are assigned the same mass value.
Making use of~(\ref{ue}), we obtain
$$\overline{{\rm E}}[f(\pmb{x},\tilde{\pmb{c}})]=\frac{2}{K(K-1)}\sum_{\{i,j\}\subseteq [K]: i\neq j}\max\{f(\pmb{x},\pmb{c}_i), f(\pmb{x},\pmb{c}_j)\}.$$
Order the scenarios so that $f(\pmb{x},\pmb{c}_1)\geq \dots\geq f(\pmb{x},\pmb{c}_K)$. Then
$$\overline{{\rm E}}[f(\pmb{x},\tilde{\pmb{c}})]=\frac{2}{K(K-1)}\left( (K-1)f(\pmb{x},\pmb{c}_1)+(K-2)f(\pmb{x},\pmb{c}_2)+\dots+(K-K) f(\pmb{x},\pmb{c}_K)\right)$$
Hence,
$$\overline{{\rm E}}[f(\pmb{x},\tilde{\pmb{c}})]={\rm Owa}_{\pmb{w}}(\pmb{x}),$$
where $\pmb{w}=(\frac{2(K-1)}{K(K-1)}, \frac{2(K-2)}{K(K-1)},\dots,0)$. 
Similarly, using~(\ref{le}) we get
$$\underline{{\rm E}}[f(\pmb{x},\tilde{\pmb{c}})]={\rm Owa}_{\pmb{w}'}(\pmb{x}),$$
where $\pmb{w}'=(0,\frac{2\cdot 1}{K(K-1)},\frac{2\cdot 2}{K(K-1)},\dots, \frac{2(K-1)}{K(K-1)})$. So, $w'_k=w_{K-k+1}$, $k\in [K]$.
It is not difficult to generalize this example, and this case is described in the following proposition:
\begin{prop}
\label{propowa}
Suppose that
$m(A)=1/{K \choose l}$ for each $A\subseteq [K]$ such that $|A|=l\leq K$
and $m(A)=0$,
 otherwise. Then $\overline{{\rm E}}[f(\pmb{x},\tilde{\pmb{c}})]$ is the OWA criterion
with weights $w_k={K-k\choose l-1}/{K \choose l}$ for $k=1,\dots, K-l+1$ and $w_k=0$ for $k=K-l+2,\dots, K$.
Similarly, $\underline{{\rm E}}[f(\pmb{x},\tilde{\pmb{c}})]$ is the OWA criterion with weights $w'_k=w_{K-k+1}$, $k\in [K]$.
\end{prop}
\begin{proof}
	The following equation
	$$\overline{{\rm E}}[f(\pmb{x},\tilde{\pmb{c}})]=1/{K \choose l}\sum_{A\subseteq [K]: |A|=l} \max_{k\in A} f(\pmb{x},\pmb{c}_k)$$
	holds.
	Let us order the scenarios so that $f(\pmb{x},\pmb{c}_1)\geq \dots\geq f(\pmb{x},\pmb{c}_K)$. There are exactly ${K-k\choose l-1}$ subsets $A\subseteq [K]$ such that $A=\{k\} \cup A'$, $|A|'=l-1$, and $A'\subseteq \{k+1,\dots, K\}$. Therefore, $f(\pmb{x},\pmb{c}_k)$ is the maximizer over $k\in A$. Thus,
	$$\overline{{\rm E}}[f(\pmb{x},\tilde{\pmb{c}})]=1/{K \choose l}\left({K-1\choose l-1} f(\pmb{x},\pmb{c}_1)+{K-2\choose l-1} f(\pmb{x},\pmb{c}_2)+\dots +{l-1\choose l-1} f(\pmb{x},\pmb{c}_{K-l+1})\right)$$
	and the proposition follows. The case for $\underline{{\rm E}}[f(\pmb{x},\tilde{\pmb{c}})]$ is just symmetric.
\end{proof}

\section{Computational complexity of the ROBF problem}
\label{seccompl}

From the known results in robust combinatorial optimization, it  immediately follows that~(\ref{pfh1}) is NP-hard when $\mathbb{X}\subseteq \{0,1\}^n$. Indeed, the min-max problem, being a special case of~(\ref{pfh1}), is already NP-hard for $K=2$ and strongly NP-hard when $K$ is unbounded for basic selection and network problems (see~\cite{KZ16b} for a survey). The following result for $\alpha\in [0,1)$ shows that~(\ref{pfh1}) can also be  NP-hard if the underlying problem~(\ref{pf}) is a 
linear programming problem (a continuous case).
We show in Section~\ref{secsol} that the case $\alpha=1$ is easier to solve.
\begin{thm}
\label{compl}
The problem~(\ref{pfh1}) is NP-hard for each $\alpha \in [0,1)$, when $\mathbb{X}$ is a polyhedron in $[0,1]^n$ and the function $f$ is linear.
\end{thm}
\begin{proof}
Consider the following \textsc{MINSAT} problem. Given a set of $s$  boolean variables $q_1,\dots q_s$, a collection of $t$ clauses $\mathcal{C}_1,\dots, \mathcal{C}_t$ over the boolean variables and a positive integer $r< t$. 
We ask if there is a truth assignment to the variables in which at least $r$ clauses are not satisfied. The \textsc{MINSAT} problem is known to be NP-complete, even if each clause contains at most two literals~\cite{KM94}.  

We now propose a  reduction from  \textsc{MINSAT} to problem~(\ref{pfh1}).
Given an instance of   \textsc{MINSAT}, we build the corresponding instance of problem~(\ref{pfh1}).
We can assume w.l.o.g. that~$t$ is even and $r=t/2$ (as we can add a valid number of dummy clauses to the instance). 
Let us define variables $x_i$ and $\overline{x}_i$ for each $i\in [s+1]$, so the number of variables is $n=2s+2$. 
Define the polyhedron $\mathbb{X}$ by using constraints $x_i\geq 0$, $\overline{x}_i\geq 0$ and $x_i+\overline{x}_i=1$, $i\in [s+1]$. For each clause $\mathcal{C}_i$ we form scenario $\pmb{c}_i$ as follows.
If $q_j\in \mathcal{C}_i$, then the cost of $x_j$ is~1; if $\overline{q}_j\in \mathcal{C}_i$
($\overline{q}_j$ is the negation of~$q_j$), then the cost of $\overline{x}_j$
is~1 under $\pmb{c}_i$; the costs of the remaining variables under $\pmb{c}_i$ are set to~0. Next, we add $r$ scenarios $\pmb{c}'_1,\dots,\pmb{c}'_r$ under which the costs of $x_{s+1}$ and $\overline{x}_{s+1}$ are~2 and the costs of the remaining variables are~0.
 Hence, the total number of scenarios formed is $K=t+r=3r$.  We choose the linear cost function $f(\pmb{x},\pmb{c})=\pmb{c}^T\pmb{x}$. Finally, the mass function $m$ assigns the mass  $1/{t+r \choose 2r+1}=1/{3r \choose 2r+1}$ to each subset of scenarios with the cardinality 
of~$2r+1$ and mass~0 to all the remaining subsets. Observe (see Proposition~\ref{propowa}), that
\begin{equation}
\label{owared}
{\rm H}_{\alpha}(\pmb{x})=\alpha {\rm Owa}_{\pmb{w}}(\pmb{x}) + (1-\alpha){\rm Owa}_{\pmb{w}'}(\pmb{x}).
\end{equation}
where
$$\pmb{w}=\left(\underbrace{\frac{{3r-1\choose 2r}}{{3r \choose 2r+1}}, \dots, \frac{{3r-r\choose 2r}}{{3r \choose 2r+1}}}_{r \;\text{positive weights}},\underbrace{0,\dots,0}_{2r\; \text{weights}} \right), \;\;  \pmb{w}'=\left(\underbrace{0,\dots,0}_{2r\; \text{weights}}, \underbrace{\frac{{3r-r\choose 2r}}{{3r \choose 2r+1}}, \dots, \frac{{3r-1\choose 2r}}{{3r \choose 2r+1}}}_{r \;\text{positive weights}}, \right)$$
Therefore, only the first $r$ weights are positive in $\pmb{w}$ and only the last $r$ weights are positive in $\pmb{w}'$. 

 To illustrate the reduction, consider a sample instance with variables $q_1,q_2,q_3,q_4$, clauses $(q_1 \vee \overline{q}_2)$,
  $(\overline{q}_2\vee q_3)$, $(q_1 \vee q_3)$, $(\overline{q}_3 \vee q_4)$, 
  $(q_1\vee \overline{q}_4)$, $(\overline{q}_2 \vee \overline{q}_4)$, $(\overline{q}_1,\vee q_2)$, $(q_1 \vee q_2)$ and $r=4$.  
  The scenarios for this instance are shown in Table~\ref{tab1}.
 \begin{center}
 \begin{table}[ht]
 	 \centering
 	 \caption{Scenarios for the sample instance of \textsc{MINSAT}.}\label{tab1}
 	\begin{tabular}{cc|cccccccccccccccc}
    	&  $\pmb{x}$ & $\pmb{c}_1$ & $\pmb{c}_2$ & $\pmb{c}_3$ & $\pmb{c}_4$ & $\pmb{c}_5$ & $\pmb{c}_6$ & $\pmb{c}_7$ & $\pmb{c}_8$& $\pmb{c}'_1$ & $\pmb{c}'_2$ & $\pmb{c}'_3$ & $\pmb{c}'_4$ \\ \hline
	0 & $x_1$ & 1 & 0 & 1 & 0 & 1 & 0 & 0 & 1 & 0& 0 & 0 & 0\\
	1 & $\overline{x}_1$ & 0 & 0 & 0 & 0 & 0 & 0 & 1 & 0 & 0 & 0 & 0 & 0\\
	1 & $x_2$ & 0 & 0 & 0 & 0 & 0 & 0 & 1 & 1 & 0 & 0 & 0 & 0\\
	0 & $\overline{x}_2$ & 1 & 1 & 0 & 0 & 0 &1 & 0 & 0 & 0 & 0 & 0 & 0 \\
	0 & $x_3$ & 0 & 1 & 1 & 0 & 0 & 0 & 0 & 0 & 0 & 0 & 0 & 0\\
	1 & $\overline{x}_3$ & 0 & 0 & 0 & 1 & 0 & 0 & 0 & 0 & 0 & 0 & 0 & 0\\
	0 & $x_4$ & 0 & 0 & 0 & 1 & 0  & 0 & 0 & 0 & 0 & 0 & 0 & 0\\
	1 & $\overline{x}_4$ & 0 & 0 & 0 & 0 & 1 & 1 & 0 & 0& 0 & 0 & 0 & 0\\ \hline
	   & $x_5$ & 0 & 0 & 0 & 0 & 0 & 0 & 0 & 0& 2 & 2 & 2 & 2\\
	   & $\overline{x}_5$ & 0 & 0 & 0 & 0 & 0 & 0 & 0 & 0& 2 & 2 & 2 & 2\\
	\end{tabular}
 \end{table}
 \end{center}
 We fix $m(A)=1/{12 \choose 9}=1/220$ for each $A\subseteq [K]$ such that $|A|=9$ and~0, otherwise. Polyhedron $\mathbb{X}$ is defined by $x_i+\overline{x}_i=1$, $x_i,\overline{x}_i\geq 0$, $i=1,\dots,5$. The weights in~(\ref{owared}) are then $\pmb{w}=(\frac{165}{220}, \frac{45}{220},\frac{9}{220},\frac{1}{220},0,0,0,0,0,0,0,0)$ and $\pmb{w}'=(0,0,0,0,0,0,0,0,\frac{1}{220}, \frac{9}{220},\frac{45}{220},\frac{165}{220})$.

Observe that in this proof, we have constructed an instance of the problem~\eqref{pfh1} with an exponential number of focal sets. However, it is possible to describe this instance  in a polynomial of the size of the instace of
\textsc{MINSAT} due to Proposition~\ref{propowa}.
Hence, the reduction is polynomial.

 We now show that the answer to \textsc{MINSAT} is yes if  and only if $H_{\alpha}(\pmb{x})\leq 2\alpha$ for some feasible solution $\pmb{x}\in\mathbb{X}$.  Because the cost of any solution $\pmb{x}$ under scenarios $\pmb{c}'_1, \dots,\pmb{c}'_r$ 
 is equal to~2 and under scenarios $\pmb{c}_1,\dots,\pmb{c}_t$ is not greater than~2, we get  ${\rm Owa}_{\pmb{w}}(\pmb{x})=2$ for each $\pmb{x}\in \mathbb{X}$. Hence, it is enough to show that  the answer to \textsc{MINSAT} is yes if  and only if ${\rm Owa}_{\pmb{w}'}(\pmb{x})=0$ for some feasible solution $\pmb{x}\in\mathbb{X}$
 
 Assume that the answer to \textsc{MINSAT} is yes and let $q_1,\dots,q_s$ be a truth assignment to the variables that satisfies no more than $r$ clauses. Let us form a feasible solution $\pmb{x}\in \mathbb{X}$ such that $x_j=1, \overline{x}_j=0$ if $q_j=1$  and $x_j=0, \overline{x}_j=1$ if $q_j=0$. We also fix $x_{s+1}=1$ and $\overline{x}_{s+1}=0$. By the construction, there are at least $r$ scenarios under which the cost of $\pmb{x}$ is~0 and thus no more than $2r$ scenarios under which the cost of $\pmb{x}$ is positive. Hence, ${\rm Owa}_{\pmb{w}'}(\pmb{x})=0$.
 
 Assume that ${\rm Owa}_{\pmb{w}'}(\pmb{x})=0$ for some feasible solution $\pmb{x}\in \mathbb{X}$.  Because the costs under scenarios are nonnegative, there must be at least $r$ scenarios, say $\pmb{c}_1,\dots, \pmb{c}_r$ under which the cost of $\pmb{x}$ is~0. We get $\pmb{c}_i^{T}\pmb{x}=0$ if and only if $x_j=0$, $\overline{x}_j=1$  ($\overline{x}_j=0$, $x_j=1$) if $q_j\in \mathcal{C}_i$ ($\overline{q}_j\in \mathcal{C}_i$). It is possible for some $j\in [s]$ that neither $x_j$ nor $\overline{x}_j$ appears in the clauses $\mathcal{C}_1,\dots, \mathcal{C}_r$ corresponding to $\pmb{c}_1,\dots, \pmb{c}_r$ (so $x_j$ can be fractional). In this case we assign any value to $q_j$ which does not change the values of $\mathcal{C}_1,\dots,\mathcal{C}_r$ (they are still not satisfied). This defines a truth assignment to $q_1,\dots, q_s$ under which at least the $r$ clauses are not satisfied.  
\end{proof}

\begin{cor}
	If $\alpha=0$, then the problem~(\ref{pfh1}) is not  approximable, if $\mathbb{X}$ is a polyhedron in $[0,1]^n$ and the function $f$ is linear, unless P=NP.
\end{cor}
\begin{proof}
	In the proof of Theorem~\ref{compl} for $\alpha=0$, the answer to MINSAT is yes if and only if there is a solution $\pmb{x}\in \mathbb{X}$ such that ${\rm H}_{\alpha}(\pmb{x})\leq 0$. Hence, any approximation algorithm would be used to solve the NP-hard MINSAT problem.
\end{proof}

\section{Solving the ROBF problem}
\label{secsol}
In this section, we propose some methods of solving~(\ref{pfh1}). We also identify some special cases of this problem which can be solved in polynomial time. Problem~(\ref{pfh1}) can be represented as the following mixed integer program:
\begin{equation}
\label{mippf}
\begin{array}{llll}
	\min & \displaystyle \sum_{F\in \mathcal{F}} m(F) y(F) \\
		& y(F)\geq \alpha f(\pmb{x},\pmb{c}_k)+(1-\alpha)u(F) & F\in \mathcal{F}, k\in F\\
		& u(F)\geq  f(\pmb{x},\pmb{c}_k)-M(1-\delta_k(F)) & F\in \mathcal{F}, k\in F\\
		& \displaystyle \sum_{k\in F} \delta_k(F)=1 & F\in \mathcal{F}\\
		& \delta_k(F)\in \{0,1\} & F\in \mathcal{F}, k\in F\\
		& y(F), u(F) \geq 0 & F\in \mathcal{F}\\
		& \pmb{x}\in \mathbb{X},
\end{array}
\end{equation}
where $M$ is a large constant. The computational complexity of solving~(\ref{mippf}) depends on the number of focal sets and the definition of the cost function $f$. In particular, when $|\mathcal{F}|$ is not very large and the function $f$ is linear, the problem can be solved using some MIP solvers.
%Observe that solve~(\ref{mippf}) by guessing a variable $\delta_k(F)$ which equals~1 for each $F$. If $\mathcal{F}=\{F_1,\dots,F_{\ell}\}$, then there are $|F_1| |F_{2}|\dots |F_{\ell}|\leq z(\mathcal{F})^{\ell}$ such a choices.
 The problem simplifies when $\alpha=1$. Using~(\ref{ue}) we immediately get the following result:
\begin{prop}
\label{prop1a}
	If $\alpha=1$, then the problem~(\ref{pfh1}) is equivalent to the following optimization problem:
\begin{equation}
\label{lpf}
\begin{array}{llll}
	\min & \displaystyle \sum_{F\in \mathcal{F}} m(F) y(F) \\
		& y(F)\geq f(\pmb{x},\pmb{c}_k) & F\in \mathcal{F}, k\in F\\
		& \pmb{x}\in \mathbb{X}
\end{array}
\end{equation}
\end{prop}
If the function $f$ is convex in $\pmb{x}$ and the set $\mathbb{X}$ is convex, then~(\ref{lpf}) is a convex optimization problem, which in most cases can be solved efficiently (see, e.g.,~\cite{BV04}). In particular, if the function $f$ is linear, then~(\ref{lpf}) is a linear programming problem.
\begin{prop}
\label{prop2a}
	If $\alpha\in [0.5,1]$ and $z(\mathcal{F})=2$,
	($z(\mathcal{F})=\max_{F\in \mathcal{F}} |F|$),
	then problem~(\ref{pfh1}) is equivalent to the following optimization problem:
\begin{equation}
\label{lpfh}
\begin{array}{llll}
	\min & \displaystyle \sum_{k\in [K]} m_k f(\pmb{x},\pmb{c}_k)+\sum_{\{i,j\}\in \mathcal{F}} m_{ij} y_{ij} \\
		& y_{ij}\geq \alpha f(\pmb{x},\pmb{c}_i)+(1-\alpha)f(\pmb{x},\pmb{c}_j) & \{i,j\}\in \mathcal{F}\\
		& y_{ij}\geq \alpha f(\pmb{x},\pmb{c}_j)+(1-\alpha)f(\pmb{x},\pmb{c}_i)& \{i,j\}\in \mathcal{F}\\
		& \pmb{x}\in \mathbb{X},
\end{array}
\end{equation}
where $m_k=m(\{k\})$, $k\in [K]$, and $m_{ij}=m(\{i,j\})$ for $\{i,j\}\in \mathcal{F}$.
\end{prop}
\begin{proof}
	From~(\ref{ue}) and~(\ref{le}), 
	 it follows that
	the objective function of~(\ref{pfh1}) can be expressed as
	$$\sum_{k\in [K]} m_k f(\pmb{x},\pmb{c}_k)+\sum_{\{i,j\}\in \mathcal{F}} m_{ij} (\alpha\max\{f(\pmb{x},\pmb{c}_i),f(\pmb{x},\pmb{c}_j)\}+(1-\alpha)\min\{f(\pmb{x},\pmb{c}_i),f(\pmb{x},\pmb{c}_j)\})=$$
	$$\sum_{k\in [K]} m_k f(\pmb{x},\pmb{c}_k)+\sum_{\{i,j\}\in \mathcal{F}} m_{ij} \max\{\alpha f(\pmb{x},\pmb{c}_i)+(1-\alpha)f(\pmb{x},\pmb{c}_j), \alpha f(\pmb{x},\pmb{c}_j)+(1-\alpha)f(\pmb{x},\pmb{c}_i)\},$$
	where the equality follows from the assumption that $\alpha\geq 0.5$. Introducing additional variables $y_{ij}$, $\{i,j\}\in \mathcal{F}$ to express the inner maxima yields~(\ref{lpfh}).
\end{proof}
Observe that if the function $f$ is convex in $\pmb{x}$, then~(\ref{lpfh}) is a convex optimization problem. In particular, if the function $f$ is linear, then it becomes a linear programming problem.

\begin{thm}
\label{thmdec}
	If $\mathbb{X}$ is a convex set in $\mathbb{R}^n$ and $f$ is a convex function, then~(\ref{pfh1}) can be solved by solving a family  $z(\mathcal{F})^{\ell}$ of convex problems, where $\ell=|\mathcal{F}|$.
\end{thm}
\begin{proof}
	Observe that for each $F\in \mathcal{F}$, there is exactly one binary variable $\delta_k(F)$, $k\in F$, in~(\ref{mippf}) which is equal to~1. Denote the index of this variable by $k_F$. If we could guess the optimal values for $k_F$, $F\in \mathcal{F}$, then the model~(\ref{mippf}) simplifies to the following convex problem:
	\begin{equation}
\label{mippf1}
\begin{array}{llll}
	\min & \displaystyle \sum_{F\in \mathcal{F}} m(F) y(F) \\
		& y(F)\geq \alpha f(\pmb{x},\pmb{c}_k)+(1-\alpha)f(\pmb{x},\pmb{c}_{k_F}) & F\in \mathcal{F}\\
		& y(F) \geq 0 & F\in \mathcal{F}\\
		& \pmb{x}\in \mathbb{X},
\end{array}
\end{equation}
There are at most $z(\mathcal{F})^{\ell}$ possible choices of  $k_F$, $F\in \mathcal{F}$ and we can try all of them to find the one that minimizes~(\ref{mippf1}).
 Hence~(\ref{mippf}) can be decomposed into solving $z(\mathcal{F})^{\ell}$ problems of the type~(\ref{mippf1}).
\end{proof}
Theorem~\ref{thmdec} can be applied when the number of focal sets in $\mathcal{F}$ is not large. In particular, if 
problem~(\ref{pf}) is a linear programming problem and $z(\mathcal{F})$ is bounded by a constant, then problem~(\ref{pfh1}) can be solved in polynomial time.

\section{Approximation algorithms for  the ROBF problem}
\label{secappr}

In this section, we will propose two approximation algorithms for~(\ref{pfh1}). In the first we will assume w.l.o.g. that  $|F|\geq 2$ for each focal set $F\in \mathcal{F}$. Indeed, if a positive mass is assigned to a single scenario $\pmb{c}_k$, then we add a copy $\pmb{c}'_k$ of this scenario to $\mathcal{U}$ and fix  $m(\{k,k'\}):=m(\{k\})$ and $m(\{k\}):=0$. This transformation does not change  the values of  $\overline{E}[f(\pmb{x},\tilde{\pmb{c}})]$ and $\underline{E}[f(\pmb{x},\tilde{\pmb{c}})]$ (see equations~(\ref{ue}) and~(\ref{le})). Define 
$$
v_k=\sum_{F\in\mathcal{F}: k\in F}  m(F), \; k\in [K].
$$
%\textcolor{red}{We may add a comment that say that $v_k$ is the plausibility measure $Pl(\{k\})$ but in the section 2 we %defined only the $Bel$ function and no the dual one $Pl$} 
If $\{F\in\mathcal{F}: k\in F\}=\emptyset$, then we set $v_k=0$. Consider the following optimization problem:
 \begin{equation}
	\label{detp}
	\begin{array}{llll}
	\min & \displaystyle \sum_{k\in [K]} v_k f(\pmb{x},\pmb{c}_k)\\
	  & \pmb{x}\in\mathbb{X}
	  \end{array}  
\end{equation}
Observe that~(\ref{detp}) is a convex optimization problem, if $\mathbb{X}$ is a convex set and $f$ is a convex function. In particular,~(\ref{detp}) is a linear programming problem if $\mathbb{X}$ is a polytope described by a system of linear constraints and $f$ is linear. 

\begin{thm}
\label{thmappr}
	Let $\hat{\pmb{x}}\in \mathbb{X}$ be an optimal solution to~(\ref{detp}) and let $\pmb{x}^*$ be an optimal solution to~(\ref{pfh1}).
	Then 
	\begin{align}
		 & {\rm H}_{\alpha}(\hat{\pmb{x}})\leq z(\mathcal{F}) {\rm H}_{\alpha}(\pmb{x}^*) & {\rm for } & \;\;\alpha\in [0.5,1] \label{appr1} \\
		  & {\rm H}_{\alpha}(\hat{\pmb{x}})\leq \frac{1-\alpha}{\alpha} z(\mathcal{F}){\rm H}_{\alpha}(\pmb{x}^*) & {\rm for } &  \;\;\alpha\in (0,0.5] \label{appr2}
	\end{align}
\end{thm}
\begin{proof}
	We will show first that for each $\alpha\in [0,1]$ the following inequality holds:
	\begin{equation}
	\label{e00}
	{\rm H}_{\alpha}(\pmb{x}^*)\geq \frac{\alpha}{z(\mathcal{F})}\sum_{F\in \mathcal{F}}m(F)\sum_{k\in F}f(\hat{\pmb{x}},\pmb{c}_k).
	\end{equation}
	Using equations~(\ref{ue}) and~(\ref{le}) we get
	$${\rm H}_{\alpha}(\pmb{x}^*)\geq \alpha \sum_{F\in \mathcal{F}} m(F)\max_{k\in F} f(\pmb{x}^*,\pmb{c}_k)\geq
	\alpha \sum_{F\in \mathcal{F}}m(F)\frac{1}{|F|}\sum_{k\in F}f(\pmb{x}^*,\pmb{c}_k)\stackrel{(a)}{\geq}$$  
	$$\alpha \sum_{F\in \mathcal{F}}m(F)\frac{1}{z(\mathcal{F})}\sum_{k\in F}f(\pmb{x}^*,\pmb{c}_k)=\frac{\alpha}{z(\mathcal{F})}\sum_{k\in [K]} \sum_{F\in \mathcal{F}: k\in F} m(F)f(\pmb{x}^*,\pmb{c}_k)= $$
	$$
	\frac{\alpha}{z(\mathcal{F})}\sum_{k\in [K]} v_k f(\pmb{x}^*,\pmb{c}_k) \stackrel{(b)}{\geq} \frac{\alpha}{z(\mathcal{F})}\sum_{k\in [K]} v_k f(\hat{\pmb{x}},\pmb{c}_k)=\frac{\alpha}{z(\mathcal{F})}\sum_{F\in \mathcal{F}}m(F)\sum_{k\in F}f(\hat{\pmb{x}},\pmb{c}_k).
	$$
Inequality~$(a)$ follows from the fact that $z(\mathcal{F})\geq |F|$ for each $F\in \mathcal{F}$. In the inequality~$(b)$ we use the optimality of $\hat{\pmb{x}}$ in~(\ref{detp}).
 If $\alpha\in [0.5,1]$ and $|F|\geq 2$, then 
 $$\alpha\sum_{k\in F}f(\hat{\pmb{x}}, \pmb{c}_k)\geq \alpha \max_{k\in F} f(\hat{\pmb{x}}, \pmb{c}_k)+(1-\alpha)\min_{k\in F} f(\hat{\pmb{x}}, \pmb{c}_k)$$ 
 and~(\ref{e00}) implies ${\rm H}_{\alpha}(\pmb{x}^*)\geq \frac{1}{z(\mathcal{F})} {\rm H}_{\alpha}(\hat{\pmb{x}})$, which yields~(\ref{appr1}).
 On the other hand, if $\alpha\in (0,0.5]$ and $|F|\geq 2$, then 
 $$(1-\alpha)\sum_{k\in F}f(\hat{\pmb{x}}, \pmb{c}_k)\geq \alpha \max_{k\in F} f(\hat{\pmb{x}}, \pmb{c}_k)+(1-\alpha)\min_{k\in F} f(\hat{\pmb{x}}, \pmb{c}_k)$$ 
 and~(\ref{e00}) implies ${\rm H}_{\alpha}(\pmb{x}^*)\geq \frac{\alpha}{(1-\alpha)z(\mathcal{F})} {\rm H}_{\alpha}(\hat{\pmb{x}})$, which yields~(\ref{appr2}).
\end{proof}

\begin{thm}\label{gammaapprox}
	If problem~(\ref{lpf}) is approximable within $\gamma\geq 1$,
	then~(\ref{pfh1}) is approximable within 
	$\gamma\left( 1+\frac{(1-\alpha)}{\alpha}
	 \frac{\sum_{F\in \mathcal{F}} \min_{k\in F} f(\hat{\pmb{x}},\pmb{c}_k)}{\sum_{F\in \mathcal{F}} \max_{k\in F} f(\hat{\pmb{x}},\pmb{c}_k)}\right)$	
	for every $\alpha\in (0,1]$, where $\hat{\pmb{x}}$ is a $\gamma$-approximate solution to~(\ref{lpf}).
\end{thm}
\begin{proof}
Let $\pmb{x}^*$ be an optimal solution to~(\ref{pfh1}) and $\hat{\pmb{x}}$ be a $\gamma$-approximate solution to~(\ref{lpf}). 
This gives a lower bound on ${\rm H}_{\alpha}(\pmb{x}^*)$, i.e.
\begin{equation}
	{\rm H}_{\alpha}(\pmb{x}^*)\geq \alpha \sum_{F\in \mathcal{F}} \max_{k\in F} f(\pmb{x}^*,\pmb{c}_k)\geq \frac{\alpha}{\gamma} \sum_{F\in \mathcal{F}} \max_{k\in F} f(\hat{\pmb{x}},\pmb{c}_k).\label{lbhal}
\end{equation}
Consider an upper bound on the relative error of ${\rm H}_{\alpha}(\hat{\pmb{x}})$:
\begin{align*}
&\frac{{\rm H}_{\alpha}(\hat{\pmb{x}})-{\rm H}_{\alpha}(\pmb{x}^*)}{{\rm H}_{\alpha}(\pmb{x}^*)}
\stackrel{(\ref{lbhal})}{\leq} \frac{{\rm H}_{\alpha}(\hat{\pmb{x}})-
\frac{\alpha}{\gamma} \sum_{F\in \mathcal{F}} \max_{k\in F} f(\hat{\pmb{x}},\pmb{c}_k)}
{\frac{\alpha}{\gamma} \sum_{F\in \mathcal{F}} \max_{k\in F} f(\hat{\pmb{x}},\pmb{c}_k)}\\
&= \frac{\alpha \sum_{F\in \mathcal{F}} \max_{k\in F} f(\hat{\pmb{x}},\pmb{c}_k)+
(1-\alpha) \sum_{F\in \mathcal{F}} \min_{k\in F} f(\hat{\pmb{x}},\pmb{c}_k)-
\frac{\alpha}{\gamma} \sum_{F\in \mathcal{F}} \max_{k\in F} f(\hat{\pmb{x}},\pmb{c}_k)}
{\frac{\alpha}{\gamma} \sum_{F\in \mathcal{F}} \max_{k\in F} f(\hat{\pmb{x}},\pmb{c}_k)}\\
&\leq \gamma\left( 1+\frac{(1-\alpha)}{\alpha}
	 \frac{\sum_{F\in \mathcal{F}} \min_{k\in F} f(\hat{\pmb{x}},\pmb{c}_k)}{\sum_{F\in \mathcal{F}} \max_{k\in F} f(\hat{\pmb{x}},\pmb{c}_k)}\right) -1.
\end{align*}
Thus, we get
\[
{\rm H}_{\alpha}(\hat{\pmb{x}})\leq \gamma\left( 1+\frac{(1-\alpha)}{\alpha}
	 \frac{\sum_{F\in \mathcal{F}} \min_{k\in F} f(\hat{\pmb{x}},\pmb{c}_k)}{\sum_{F\in \mathcal{F}} \max_{k\in F} f(\hat{\pmb{x}},\pmb{c}_k)}\right){\rm H}_{\alpha}(\pmb{x}^*),
\]
where $\frac{\sum_{F\in \mathcal{F}} \min_{k\in F} f(\hat{\pmb{x}},\pmb{c}_k)}{\sum_{F\in \mathcal{F}} \max_{k\in F} f(\hat{\pmb{x}},\pmb{c}_k)}\in [0,1]$.
\end{proof}

Observe that~(\ref{lpf}) is a special case of the mixed-uncertainty robust optimization problem discussed in~\cite{dokka2020mixed}, where uncertainty sets are given and we would like to solve
\[ \min_{\pmb{x}\in\mathbb{X}} \sum_{F\in\mathcal{F}} m(F) \max_{k\in F} f(\pmb{x},\pmb{c}_k), \]
 In~\cite{dokka2020mixed} it has been shown that solving the nominal problem $\min_{\pmb{x}\in\mathbb{X}} f(\pmb{x},\hat{\pmb{c}})$ with the cost verctor $\hat{\pmb{c}} = \sum_{F\in\mathcal{F}} m(F)\frac{1}{|F|}\sum_{k\in F} \pmb{c}_k$, for linear cost functions $f$, gives a $z(\mathcal{F})$-approximation to~(\ref{lpf}). In combination with Theorem~\ref{gammaapprox}, this allows us to derive approximation algorithms for a wide range of 
the ROBF problems if the underling deterministic problem~(\ref{pf}) is polynomially solvable, in particular for a wide class of network problems (see~\cite{AMO93} for a comprehensive review). Note also that in most cases $\gamma=1$ if~(\ref{lpf}) is a convex problem (in particular linear programing one).

\section{Extension to possibilistic optimization}
\label{secpos}

In this section, we will show how to use possibility theory (see, e.g.,~\cite{DP88}) to extend the model of uncertainty described in Section~\ref{secrob}. We will follow the reasoning of~\cite{TD21}. In the uncertainty representation used
 in Section~\ref{secrob}, each focal set $F \in \mathcal{F}$ can be described by a characteristic function $\pi_F:[K]\rightarrow \{0,1\}$ such that $\pi_{F}(u)=1$ if $u\in F$ and $\pi_{F}(u)=0$ if $u\notin F$. We can interpret $\pi_{F}$ as a \emph{possibility distribution} for  evidence $F$, i.e scenarios whose indices are in $F$ are possible whereas scenarios with indices not in $F$ are impossible.
 The possibility distribution $\pi_{F}$ induces the following possibility and necessity measures on~$2^{[K]}$:
\begin{equation}
\label{posdistr}
\Pi_{F}(A)=\max_{u\in A}\pi_{F}(u)
\end{equation}
\begin{equation}
\label{necdistr}
{\rm N}_{F}(A)=1-\Pi_{F}(A^c)=\min_{u\notin A}(1-\pi_{F}(u)),
\end{equation}
where the set $A^c=[K]\setminus A$ is the complement of $A$. 
The belief function~(\ref{belf}) can be alternatively represented as follows:
$${\rm Bel}(A)=\sum_{F\in \mathcal{F}} m(F){\rm N}_{F}(A),$$
which easily follows from the fact that ${\rm N}_{F}(A)=1$ if $F\subseteq A$ and ${\rm N}_{F}(A)=0$, otherwise.

A generalization of the model (see~\cite{TD21}) consists in allowing an evidence to
be represented by a 
\emph{fuzzy set} $\widetilde{F}$ in $[K]$
with the membership function $\mu_{\mathcal{\widetilde{F}}}:[K]\rightarrow [0,1]$.  We only assume that $\widetilde{F}$ is a normal fuzzy set, i.e. $\mu_{\mathcal{\widetilde{F}}}(u)=1$ for some $u\in [K]$.
Now any possibility distribution~$\pi_{\mathcal{\widetilde{F}}}$,
$\pi_{\mathcal{\widetilde{F}}}=\mu_{\mathcal{\widetilde{F}}}$,
can be used to characterize the evidence $\widetilde{F}$ and $\widetilde{\mathcal{F}}=\{\widetilde{F}_1,\dots,\widetilde{F}_{\ell}\}$ is a family of fuzzy sets in $[K]$ 
with membership functions~$\mu_{\mathcal{\widetilde{F}}_i}$ (possibility distributions~$\pi_{\mathcal{\widetilde{F}}_i}$), $i\in [\ell]$. A mass $m(\widetilde{F}_i)\in (0,1]$ is specified for each fuzzy focal set $\widetilde{F}_i$ and  is such that $\sum_{i\in [\ell]} m(\widetilde{F}_i)=1$. The family $\widetilde{\mathcal{F}}$ with the mass function~$m$ can be seen as a \emph{random fuzzy set}.

We can provide the following interpretation of this uncertainty representation. The fuzzy focal set $\widetilde{F}_i$ is some evidence, which the decision maker knows about scenarios in $\mathcal{U}$ (it represents a knowledge possessed in the scenario set). As previously, the mass $m(\widetilde{F}_i)$ supports this evidence and the mass function $m$ can be interpreted as a subjective probability distribution in $\mathcal{F}$ .  For a detailed interpretation of the model, we refer the reader to~\cite{TD21}. We can now extend the definition of the belief function in the following way:
\begin{equation}
\label{necdistr1}
{\rm Bel}(A)=\sum_{\widetilde{F}\in \mathcal{\widetilde{F}}} m(\widetilde{F}){\rm N}_{\widetilde{F}}(A)=\sum_{\widetilde{F}\in \mathcal{\widetilde{F}}} m(\widetilde{F})\min_{u\notin A}(1-\pi_{\tilde{F}_i}(u)).
\end{equation}
Hence, the belief of $A$ can be seen as the expected necessity of $A$~\cite{TD21}.

As previously,  $\mathcal{P}(m)\subseteq [0,1]^K$ is the set of all probability distributions in $[K]$ compatible with $m$, i.e. ${\rm P}(A)\geq {\rm Bel}(A)$ for each $A\subseteq [K]$. Set $\mathcal{P}(m)$ is characterized by the system of linear constraints~(\ref{consp}) with the belief function defined as~(\ref{necdistr1}).
 As ${\rm Bel}(A)$ is constant for each 
 $A\subseteq [K]$, we conclude that $\mathcal{P}(m)$ is a polytope in $[0,1]^K$.  
 The next proposition shows that there is at least one probability distribution in $\mathcal{P}(m)$.
\begin{prop}
\label{propnon}
	The set $\mathcal{P}(m)$ is nonempty.
\end{prop}
\begin{proof}
See Appendix~\ref{app:b}.
\end{proof}

In the following, we will show that the possibilistic model of uncertainty can be reduced to the model described in Section~\ref{secrob}. Therefore,  all the results shown in Sections~\ref{secsol} and~\ref{secappr} also apply to it.
Consider a fuzzy set $\tilde{F}$ in $[K]$. Let $0<\alpha_1<\dots<\alpha_{f}=1$ be an ordered sequence of increasing numbers such that the sets ($\alpha_i$-cuts or $\alpha_i$-level sets)
$$\tilde{F}^{[\alpha_i]}=\{k\in [K]: \pi_{\tilde{F}}(k)\geq \alpha_i\}$$
form a family of nested sets
$$\widetilde{F}^{[\alpha_1]}\supset \widetilde{F}^{[\alpha_2]}\supset \cdots\supset \widetilde{F}^{[\alpha_f]}.$$
Define $m_{\tilde{F}}(\widetilde{F}^{[\alpha_1]})=\alpha_1$,  $m_{\tilde{F}}(\widetilde{F}^{[\alpha_i]})=\alpha_i-\alpha_{i-1}$ for $i=2,\dots,f$ and $m_{\tilde{F}}(B)=0$ for all the remaining subsets of $[K]$. Observe that 
\begin{equation}
\label{propm}
\sum_{A\subseteq [K]} m_{\tilde{F}}(A)=\sum_{i=1}^{f} m_{\tilde{F}}(F^{[\alpha_i]})=1.
\end{equation}

\begin{prop}
\label{prop2}
For any subset $A\subseteq [K]$ the following equation is true:
\begin{equation}
\label{nrepr}
{\rm N}_{\tilde{F}}(A)=\sum_{B\subseteq A} m_{\tilde{F}}(B).
\end{equation}
\end{prop}
\begin{proof}
If $\tilde{F}^{[\alpha_i]}\not\subseteq A$ for each $i\in [f]$, then ${\rm N}_{\tilde{F}}(A)=\min_{u\notin A}(1-\pi_{\widetilde{F}}(u))=0$ and~(\ref{nrepr}) is true. If $\tilde{F}^{[\alpha_i]}\subseteq A$ for each $i\in [f]$, then ${\rm N}_{\tilde{F}}(A)=\min_{u\notin A}(1-\pi_{\widetilde{F}}(u))=1$  and~(\ref{nrepr}) is also true.
Let $j=2,\dots,f$ be the smallest number such that $\tilde{F}^{[\alpha_j]}\subseteq\dots \subseteq \tilde{F}^{[\alpha_f]}\subseteq A$. Then,
$${\rm N}_{\tilde{F}}(A)=\min_{u\notin A}(1-\pi_{\tilde{F}}(u))=1-\alpha_{j-1}.$$
On the other hand,
$$\sum_{B\subseteq A} m_{\tilde{F}}(B)=\sum_{i=j}^f m_{\tilde{F}}(\widetilde{F}^{[\alpha_i]})=(1-\alpha_{f-1})+(\alpha_{f-1}-\alpha_{f-2})+\dots+(\alpha_{j}-\alpha_{j-1})=1-\alpha_{j-1},$$
which implies~(\ref{nrepr}).
\end{proof}

Using Proposition~\ref{prop2} we can represent~(\ref{necdistr1}) in the following way:

\begin{equation}
\label{necdistr2}
{\rm Bel}(A)=\sum_{\widetilde{F}\in \mathcal{\widetilde{F}}} m(\widetilde{F}){\rm N}_{\widetilde{F}}(A)=\sum_{\widetilde{F}\in \mathcal{\widetilde{F}}} m(\widetilde{F})\sum_{B\subseteq A} m_{\tilde{F}}(B)=\sum_{B\subseteq A}  \sum_{\widetilde{F}\in \mathcal{\widetilde{F}}} m(\widetilde{F}) m_{\tilde{F}}(B).
\end{equation}
Define 
$$m'(A)=\sum_{\widetilde{F}\in \mathcal{\widetilde{F}}} m(\widetilde{F}) m_{\tilde{F}}(A), \; A\subseteq [K].$$
Using~(\ref{propm}) and the fact that $m$ is a mass function in $\mathcal{F}$, we conclude that $m'(A)\geq 0$ for any $A\subseteq [K]$ and $\sum_{A\subseteq [K]} m'(A)=1$. Therefore, $m'$ is a mass function defined in  $2^{[K]}$ and the belief function~(\ref{necdistr1}) can be equivalently defined as
$${\rm Bel}(A)=\sum_{B\subseteq A} m'(A).$$
This fact allows us to use all the results from Sections~\ref{secsol} and~\ref{secappr} for the possibilistic model of uncertainty. Notice also that the size of the focal set $\mathcal{F}'$ for the new mass function $m'$ is at most $K\ell$, so it is polynomially bounded by $K$ and $\ell$.

\begin{figure}[h!t]
\centering
\includegraphics{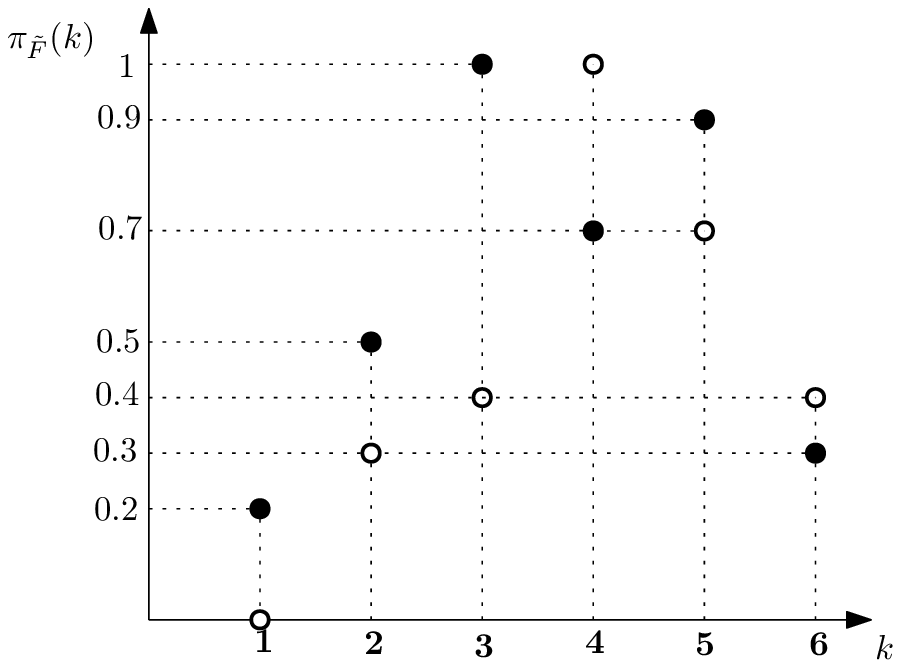}
\caption{Two fuzzy sets with $\pi_{\tilde{F}_1}: 1 \mapsto 0.2$, $2\mapsto 0.5$, $3\mapsto 1$, $4\mapsto 0.7$, $5\mapsto 0.9$, $6\mapsto 0.3$ (filled circles) and $\pi_{\tilde{F}_2}: 1\mapsto 0$, $2\mapsto 0.3$, $3\mapsto 0.4$, $4\mapsto 1$, $5\mapsto 0.7$, $6\mapsto 0.4$ (blank circles) for $K=6$ scenarios.} \label{figsamppos}
% 	\caption{Two fuzzy sets $\tilde{F}_1=\{\pmb{1}/0.2, \pmb{2}/0.5, \pmb{3}/1, \pmb{4}/0.7, \pmb{5}/0.9, \pmb{6}/0.3\}$ (filled circles) and $\tilde{F}_2=\{\pmb{1}/0, \pmb{2}/0.3, \pmb{3}/0.4, \pmb{4}/1, \pmb{5}/0.7, \pmb{6}/0.4\}$ (blank circles) for $K=6$ scenarios.} \label{figsamppos}
\end{figure}

Consider a sample set $\tilde{\mathcal{F}}=\{\tilde{F}_1, \tilde{F}_2\}$ with $m(\tilde{F}_1)=0.4$, $m(\tilde{F}_2)=0.6$ for $K=6$ scenarios shown in Figure~\ref{figsamppos}. For the fuzzy set $\tilde{F}_1$, we form the sequence $0<0.2<0.3<0.5<0.7<0.9<1$, which yields $\mathcal{F}^{[0.2]}=\{1,2,3,4,5,6\}$, $\mathcal{F}^{[0.3]}=\{2,3,4,5,6\}$, $\mathcal{F}^{[0.5]}=\{2,3,4,5\}$, $\mathcal{F}^{[0.7]}=\{3,4,5\}$, $\mathcal{F}^{[0.9]}=\{3,5\}$, $\mathcal{F}^{[1]}=\{3\}$. For the fuzzy set $\tilde{F}_2$ the corresponding sequence is $0<0.3<0.4<0.7<1$ and $\mathcal{F}^{[0.3]}=\{2,3,4,5,6\}$, $\mathcal{F}^{[0.4]}=\{3,4,5,6\}$, $\mathcal{F}^{[0.7]}=\{4,5\}$ and $\mathcal{F}^{[1]}=\{4\}$. The mass function $m'$ for $\tilde{\mathcal{F}}$ is shown in Table~\ref{tab2}.

\begin{table}[ht]
\centering
\caption{The mass function $m'(A)$ for $\tilde{\mathcal{F}}=\{\tilde{F}_1, \tilde{F}_2\}$.} \label{tab2}
\begin{tabular}{llllll}
$A$ & $m_{\tilde{F}_1}(A)$ & $m_{\tilde{F}_2}(A)$ & $m'(A)$\\ \hline
$\{1,2,3,4,5,6\}$  & 0.2 & 0 & 0.08\\ 
$\{2,3,4,5,6\}$  & 0.1 & 0.3 & 0.22\\
$\{2,3,4,5\}$ & 0.2 & 0 & 0.08\\
$\{3,4,5\}$  & 0.2  & 0 & 0.08\\
$\{3,5\}$ & 0.2 & 0 & 0.08\\
$\{3\}$ & 0.1 & 0 & 0.04\\
$\{3,4,5,6\}$ & 0 & 0.1 & 0.06\\
$\{4,5\}$ & 0 & 0.3 & 0.18\\
$\{4\}$ & 0 & 0.3 & 0.18
\end{tabular}
\end{table}

\section{Conclusions}

In this paper, we discussed a new approach to optimization under uncertainty in the objective function. We assumed that the uncertainty is described by a belief function, which is based on masses assigned to subsets of scenarios. The masses represent subjective probabilities of the respective events (subsets of scenarios) and can be provided by the 
decision maker to take into account additional knowledge about the problem parameters. The belief function induces a set of probability distributions such that the probability of each event is lower bounded by its belief. We used the generalized Hurwicz criterion to compute a solution. This criterion allows us to establish a trade-off between the worst-case and the best-case performance of solutions over all  probability distributions compatible with the belief function. The worst-case performance is consistent with  distributionally robust optimization, which has attracted considerable attention recently.
For some specific choices of the mass function, our approach leads to optimizing the OWA criterion, which has been well studied on its own.

In this paper, we identified a special case of the uncertain linear programming problem which is NP-hard and also hard to approximate. On the other hand, we described some cases of the  problem considered which can be solved or approximated in polynomial time. Finally, we extended our setting using fuzzy sets with a possibilistic interpretation. In this case, scenarios belong to focal sets with some degrees between zero and one. We showed that the fuzzy case can be reduced in polynomial time to the previous, traditional one,  and all the results for the traditional case can still be applied.

\subsubsection*{Acknowledgements}
Marc Goerigk was partially supported by the Deutsche Forschungsgemeinschaft (DFG), grant GO 2069/2-1.
Adam Kasperski and Pawe{\l} Zieli{\'n}ski were supported by
 the National Science Centre, Poland, grant 2022/45/B/HS4/00355.

%\bibliographystyle{abbrv}
%\bibliography{robust}

\begin{thebibliography}{10}

\bibitem{AMO93}
R.~K. Ahuja, T.~L. Magnanti, and J.~B. Orlin.
\newblock {\em Network {F}lows: theory, algorithms, and applications}.
\newblock Prentice Hall, Englewood Cliffs, New Jersey, 1993.

\bibitem{BN09}
A.~Ben-Tal, L.~El~Ghaoui, and A.~Nemirovski.
\newblock {\em Robust optimization}.
\newblock Princeton Series in Applied Mathematics. Princeton University Press,
  Princeton, NJ, 2009.

\bibitem{BV04}
S.~Boyd and L.~Vandenberghe.
\newblock {\em Convex optimization}.
\newblock Cambridge University Press, 2004.

\bibitem{CG15}
A.~B. Chassein and M.~Goerigk.
\newblock Alternative formulations for the ordered weighted averaging
  objective.
\newblock {\em Information Processing Letters}, 115:604--608, 2015.

\bibitem{DY10}
E.~Delage and Y.~Ye.
\newblock Distributionally robust optimization under moment uncertainty with
  application to data-deriven problems.
\newblock {\em Operations Research}, 58:595--612, 2010.

\bibitem{D67}
A.~Dempster.
\newblock Upper and lower probabilities induced by a multivalued mapping.
\newblock {\em The Annals of Math. Stat.}, 38:325--339, 1967.

\bibitem{TD20}
T.~Denoeux.
\newblock Decision-making with belief functions: A review.
\newblock {\em International Journal of Approximate Reasoning}, 109:87--110,
  2019.

\bibitem{TD21}
T.~Denoeux.
\newblock Belief functions induced by random fuzzy sets: A general framework
  for representing uncertain and fuzzy evidence.
\newblock {\em Fuzzy Sets and Systems}, 424:63--91, 2021.

\bibitem{dokka2020mixed}
T.~Dokka, M.~Goerigk, and R.~Roy.
\newblock Mixed uncertainty sets for robust combinatorial optimization.
\newblock {\em Optimization Letters}, 14(6):1323--1337, 2020.

\bibitem{DTH88}
T.~Driessen.
\newblock {\em Cooperative games, solutions and applications}.
\newblock Kluwer Academic Publishers, 1988.

\bibitem{DP88}
D.~Dubois and H.~Prade.
\newblock {\em Possibility theory: an approach to computerized processing of
  uncertainty}.
\newblock Plenum Press, New York, 1988.

\bibitem{DP90consonant}
D.~Dubois and H.~Prade.
\newblock Consonant approximations of belief functions.
\newblock {\em International Journal of Approximate Reasoning},
  4(5-6):419--449, 1990.

\bibitem{GKZ21}
R.~Guillaume, A.~Kasperski, and P.~Zieli{\'n}ski.
\newblock Robust optimization with scenarios using random fuzzy sets.
\newblock In {\em 2021 IEEE International Conference on Fuzzy Systems
  (FUZZ-IEEE)}, pages 1--6, 2021.

\bibitem{GKZ22}
R.~Guillaume, A.~Kasperski, and P.~Zieli{\'{n}}ski.
\newblock Robust optimization with scenarios using belief functions.
\newblock In N.~Trautmann and M.~Gn{\"a}gi, editors, {\em Operations Research
  Proceedings 2021}, pages 114--119, Cham, 2022. Springer International
  Publishing.

\bibitem{KZ15}
A.~Kasperski and P.~Zieli{\'n}ski.
\newblock Combinatorial optimization problems with uncertain costs and the
  {OWA} criterion.
\newblock {\em Theoretical Computer Science}, 565:102--112, 2015.

\bibitem{KZ16b}
A.~Kasperski and P.~Zieli{\'n}ski.
\newblock Robust {D}iscrete {O}ptimization {U}nder {D}iscrete and {I}nterval
  {U}ncertainty: {A} {S}urvey.
\newblock In {\em Robustness {A}nalysis in {D}ecision {A}iding, {O}ptimization,
  and {A}nalytics}, pages 113--143. Springer-Verlag, 2016.

\bibitem{KM94}
R.~Kohli, R.~Krishnamurti, and P.~Mirchandani.
\newblock The minimum satisfiability problem.
\newblock {\em SIAM Journal on Discrete Mathematics}, 7:275--283, 1994.

\bibitem{KY97}
P.~Kouvelis and G.~Yu.
\newblock {\em Robust Discrete Optimization and its Applications}.
\newblock Kluwer Academic Publishers, 1997.

\bibitem{O12}
W.~Ogryczak.
\newblock {R}obust {D}ecisions under {R}isk for {I}mprecise {P}robabilities.
\newblock In Y.~Ermoliev, M.~Makowski, and K.~Marti, editors, {\em {M}anaging
  {S}afety of {H}eterogeneous {S}ystems}, pages 51--66. Springer-Verlag, 2012.

\bibitem{OO12}
W.~Ogryczak and P.~Olender.
\newblock On {MILP} models for the {OWA} optimization.
\newblock {\em Journal of Telecommunications and Information Technology},
  2:5--12, 2012.

\bibitem{SC67}
H.~E. Scarf.
\newblock The core of an n person game.
\newblock {\em Econometrica}, 35:50--69, 1967.

\bibitem{SH76}
G.~Shafer.
\newblock {\em A mathematical theory of evidence}.
\newblock Princt. Univ. Press, 1976.

\bibitem{WKS14}
W.~Wiesemann, D.~Kuhn, and M.~Sim.
\newblock Distributionally robust convex optimization.
\newblock {\em Operations Research}, 62:1358--1376, 2014.

\bibitem{YA88}
R.~R. Yager.
\newblock On ordered weighted averaging aggregation operators in multi-criteria
  decision making.
\newblock {\em IEEE Transactions on Systems, Man and Cybernetics}, 18:183--190,
  1988.

\end{thebibliography}

\appendix

\section{Proof of equations~(\ref{ue}) and~(\ref{le})}\label{app:a}
\begin{proof}
Any element ${\rm P}$ of $\mathcal{P}(m)$ can be obtained as follows (see~\cite{D67, TD20}). Let $a(k,F)\geq 0$, $k\in [K]$, $F\in \mathcal{F}$, be nonnegative numbers satisfying the following system of equations:
\begin{equation}
\label{e1}
\sum_{k \in F} a(k,F)=m(F),\; F\in \mathcal{F}.
\end{equation}
Hence $a(k,F)$ can be interpreted as a part of the mass $m(F)$ of a focal set $F\in \mathcal{F}$  assigned to scenario $k\in [K]$.
 Then
\begin{equation}
\label{e2}
p_k=\sum_{\{F\in \mathcal{F}:\;k \in F\}} a(k,F),\; k \in [K].
\end{equation}

Using equations~(\ref{e1}) and~(\ref{e2}), we can represent the left-hand side of~(\ref{ue}) as the following linear program:
$$
	\begin{array}{lll}
	\max  & \displaystyle \sum_{k\in [K]}  \sum_{\{F\in \mathcal{F}: k \in F\}} a(k,F) f(\pmb{x},\pmb{c}_k)\\
		 & \displaystyle \sum_{k \in F} a(k,F)=m(F) & F\in \mathcal{F} \\
		 & a(k,F)\geq 0 & F\in \mathcal{F}, k\in [K]
\end{array}
$$
After exchanging the summation in the objective, we get the equivalent formulation:
\begin{equation}
\label{ee0}
	\begin{array}{lll}
	\max  & \displaystyle  \sum_{F\in \mathcal{F}} \sum_{k\in F}  a(k,F) f(\pmb{x},\pmb{c}_k)\\
		 & \displaystyle \sum_{k \in F} a(k,F)=m(F) & F\in \mathcal{F}\\
		 & a(k,F)\geq 0 & F\in \mathcal{F}, k\in [K]
\end{array}
\end{equation}
It is easy to verify that in an optimal solution to~(\ref{ee0}) one can fix $a(k^*,F)=m(F)$ for $k^*=\arg \max_{k\in F} f(\pmb{x},\pmb{c}_k)$ and  $a(k,F)=0$ for $k\in F\setminus \{k^*\}$ for each $F\in \mathcal{F}$. Hence,~(\ref{ee0}) is equivalent to the right-hand side of~(\ref{ue}). The proof of~(\ref{le}) is analogous.
\end{proof}

\section{Proof of Proposition~\ref{propnon}}\label{app:b}
\begin{proof}
We show first that the necessity measure is \emph{supermodular}, i.e.
\[
{\rm N}(A \cup B)+{\rm N}(A \cap B)\geq {\rm N}(A)+{\rm N}(B)
\]
for any $A,B\subseteq [K]$. Indeed ${\rm N}(A \cup B)\geq \max\{{\rm N}(A), {\rm N}(B)\}$, since
$1-\max_{u\in (A\cup B)^{c}}\pi(u)\geq 1-\max_{u\in A^{c}}\pi(u)$ and 
$1-\max_{u\in (A\cup B)^{c}}\pi(u)\geq 1-\max_{u\in B^{c}}\pi(u)$. Hence,
${\rm N}(A \cup B)+{\rm N}(A \cap B)\geq \max\{{\rm N}(A), {\rm N}(B)\}+\min\{{\rm N}(A), {\rm N}(B)\}={\rm N}(A)+{\rm N}(B)$. The equation ${\rm N}(A \cap B)=\min\{{\rm N}(A), {\rm N}(B)\}$ follows from the \emph{minitivity axiom} which the necessity measure satisfies~\cite{DP88}.
The supermodularity can be stated equivalently (see, e.g.,~\cite{DTH88}) as
\begin{equation}
\label{esup}
{\rm N}(A\cup\{i\})-{\rm N}(A)\geq {\rm N}(B\cup\{i\})-{\rm N}(B)
\end{equation}
for each $B\subseteq A$ and $i\notin A$. Now it is easy to see that the belief function ${\rm Bel}$ is also supermodular, as
${\rm Bel}(A\cup\{i\})-{\rm Bel}(A)\geq {\rm Bel}(B\cup\{i\})-{\rm Bel}(B)$
for each $B\subseteq A$ and $i\notin A$. This inequality follows immediately from~(\ref{esup}) and the fact that the belief function is a conic (nonnegative) linear combination of necessity measures.
The constraints~(\ref{consp}) for $\mathcal{P}(m)$ describe a core of a convex game with player set $[K]$ and characteristic function ${\rm Bel}$. The core (and thus $\mathcal{P}(m)$) of such games is nonempty (see, e.g.,~\cite{SC67}).
\end{proof}

\end{document}